\documentclass[lettersize,journal]{IEEEtran}
\usepackage{amsmath,amsfonts}
\DeclareMathOperator{\diag}{diag}
\usepackage{empheq}
\usepackage{amsthm}

\usepackage{amssymb}
\usepackage{mathtools}  
\usepackage{float}
\usepackage{algorithmic}
\usepackage{array}
\usepackage{textcomp}
\usepackage{stfloats}
\usepackage{comment}
\usepackage{url}
\usepackage{verbatim}
\usepackage{graphicx}
\usepackage{subcaption}            
\usepackage{stfloats}              
\newtheorem{theorem}{Theorem}
\newtheorem{lemma}[theorem]{Lemma}
\newtheorem{remark}{Remark}
\begin{document}
\title{On the static and small signal analysis of DAB converter}
\author{Yuxin Yang  \quad Hang Zhou \quad Hourong Song \quad Branislav Hredzak
\thanks{This paper is to be submitted to IEEE Transactions on Power Electronics Letter}}

\markboth{Letter IEEE transactions on power electronics}%
{}

\maketitle

\begin{abstract}
This document develops a method to solve the periodic operating point of Dual-Active-Bridge (DAB).
\end{abstract}


\section{Introduction}
\IEEEPARstart{T}{he} Dual-Active-Bridge is widely used for DC nano grid. In order to optimize the parameter design, a conversion ratio and a modal analysis are required. The traditional voltage \&second balancing method is not applicable.
Discrete time modeling method has been utilized to model the stability \cite{tong}\cite{lingshi}. However, the complexity of existing discrete time model of DAB converter is too high for practical use. This paper derive a simplified discrete time model without sacrificing the accuracy. Therefore it is more practical for engineering use.
\section{Problem Statement and Objectives}

We consider a piecewise linear time‐invariant system over one period \(T_s\), divided into \(n\) consecutive intervals of durations \(T_1,\dots,T_n\) (with \(T_s=\sum_{i=1}^n T_i\)).  On interval \(i\) the system is governed by
\[
\dot{\mathbf X}(t)
= A_i\,\mathbf X(t) + B_i\,\mathbf U,
\]
where the input \(\mathbf U\) is held constant in each interval.  Denote the state at the start of the \(i\)th interval by \(\mathbf X_{i-1}\), then the discrete‐time update is
\begin{equation}
    \begin{aligned}
        &\mathbf X_i
        = \Phi_i\,\mathbf X_{i-1} \;+\;\Gamma_i,
        \quad
        \Phi_i = e^{A_iT_i},\\
        &\Gamma_i = \int_0^{T_i} e^{A_i(T_i-\tau)}\,B_i\,\mathbf U\,d\tau.
    \end{aligned}
\end{equation}
Our goal is to derive:
\begin{enumerate}
  \item A closed‐form expression for \(\mathbf X_n\) in terms of \(\mathbf X_0\) and the \(\Gamma_i\).
  \item The fixed‐point equation for the periodic steady state \(\mathbf X^*\) satisfying \(\mathbf X_n=\mathbf X_0=\mathbf X^*\).
\end{enumerate}
\section{Product Notation: Direction and Boundary}

We introduce two notations for multiplying the transition matrices:
\[
\prod_{j=a}^{b}\!{}^{\rightarrow}\Phi_j
\;:=\;\Phi_a\,\Phi_{a+1}\,\cdots\,\Phi_b,
\quad
\prod_{j=a}^{b}\!{}^{\leftarrow}\Phi_j
\;:=\;\Phi_b\,\Phi_{b-1}\,\cdots\,\Phi_a,
\]
for \(a\le b\).  In particular,
\[
\prod_{j=a}^{a}\!{}^{\rightarrow}\Phi_j
= \prod_{j=a}^{a}\!{}^{\leftarrow}\Phi_j
= \Phi_a.
\]
In our closed‐form formula only the “reverse” product \(\prod_{j=1}^{n}{}^{\leftarrow}\Phi_j\) appears, avoiding any need for an empty‐product convention.

\section{Main Results: Recursive Closed‐Form and Fixed‐Point Equation}

\subsection{Recursive Closed‐Form}

For any \(n\ge1\), the state at the end of the \(n\)th interval is
\[
\boxed{
\mathbf X_n
= \Bigl(\prod_{j=1}^{n}{}^{\leftarrow}\Phi_j\Bigr)\,\mathbf X_0
+ \sum_{i=1}^{n-1}
  \Bigl(\prod_{j=i+1}^{n}{}^{\leftarrow}\Phi_j\Bigr)\,\Gamma_i
+ \Gamma_n
}
\tag{1}
\]

\subsection{Periodic Fixed‐Point Equation}

If a periodic steady state \(\mathbf X^*=\mathbf X_0=\mathbf X_n\) exists, it satisfies
\[
\boxed{
\bigl(I - \Pi\bigr)\,\mathbf X^*
= \sum_{i=1}^{n-1}
  \Bigl(\prod_{j=i+1}^{\,n}{}^{\leftarrow}\Phi_j\Bigr)\,\Gamma_i
+ \Gamma_n,
}
\quad
\Pi := \prod_{j=1}^{n}{}^{\leftarrow}\Phi_j.
\tag{2}
\]
Take four state transition as an example:
\begin{equation}
\begin{aligned}
X_1 &= \Phi_1 X_0 + \Gamma_1,\\
X_2 &= \Phi_2\Phi_1 X_0 + \Phi_2\,\Gamma_1 + \Gamma_2,\\
X_3 &= \Phi_3\Phi_2\Phi_1 X_0 + \Phi_3\Phi_2\,\Gamma_1 + \Phi_3\,\Gamma_2 + \Gamma_3,\\
X_4 &= \Phi_4\Phi_3\Phi_2\Phi_1 X_0 
      + \Phi_4\Phi_3\Phi_2\,\Gamma_1 
      + \Phi_4\Phi_3\,\Gamma_2 
      + \Phi_4\,\Gamma_3 
      + \Gamma_4.
\end{aligned}
\end{equation}

\section{System Matrices For Dual Active Bridge}
\begin{figure}[H]
  \centering
  \includegraphics[width=0.5\textwidth]{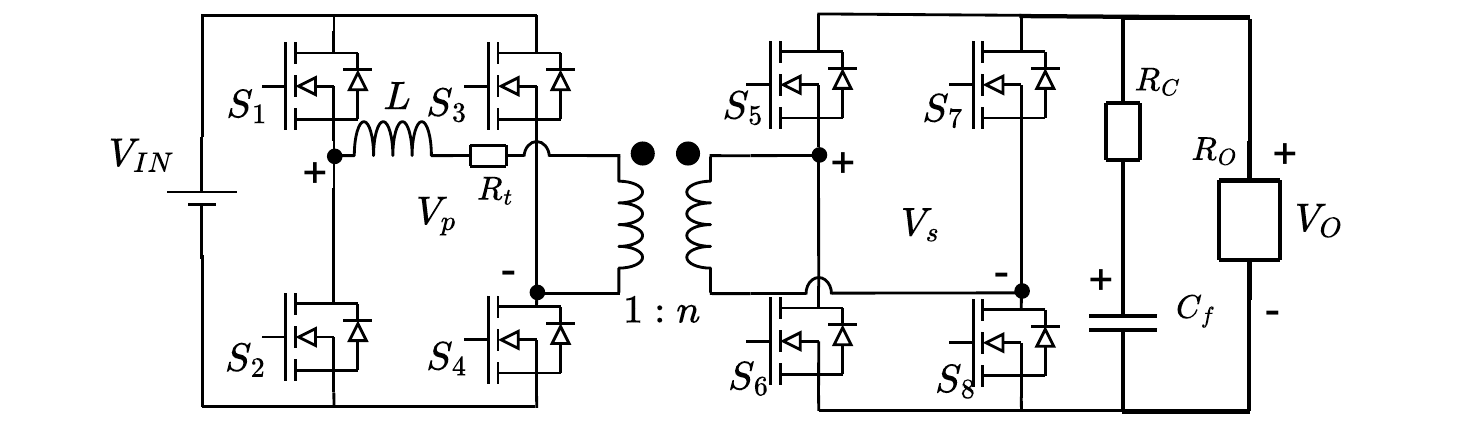}
  \caption{Schematic of the DAB converter}
\end{figure}
\begin{figure}[H]
  \centering
  \includegraphics[width=0.5\textwidth]{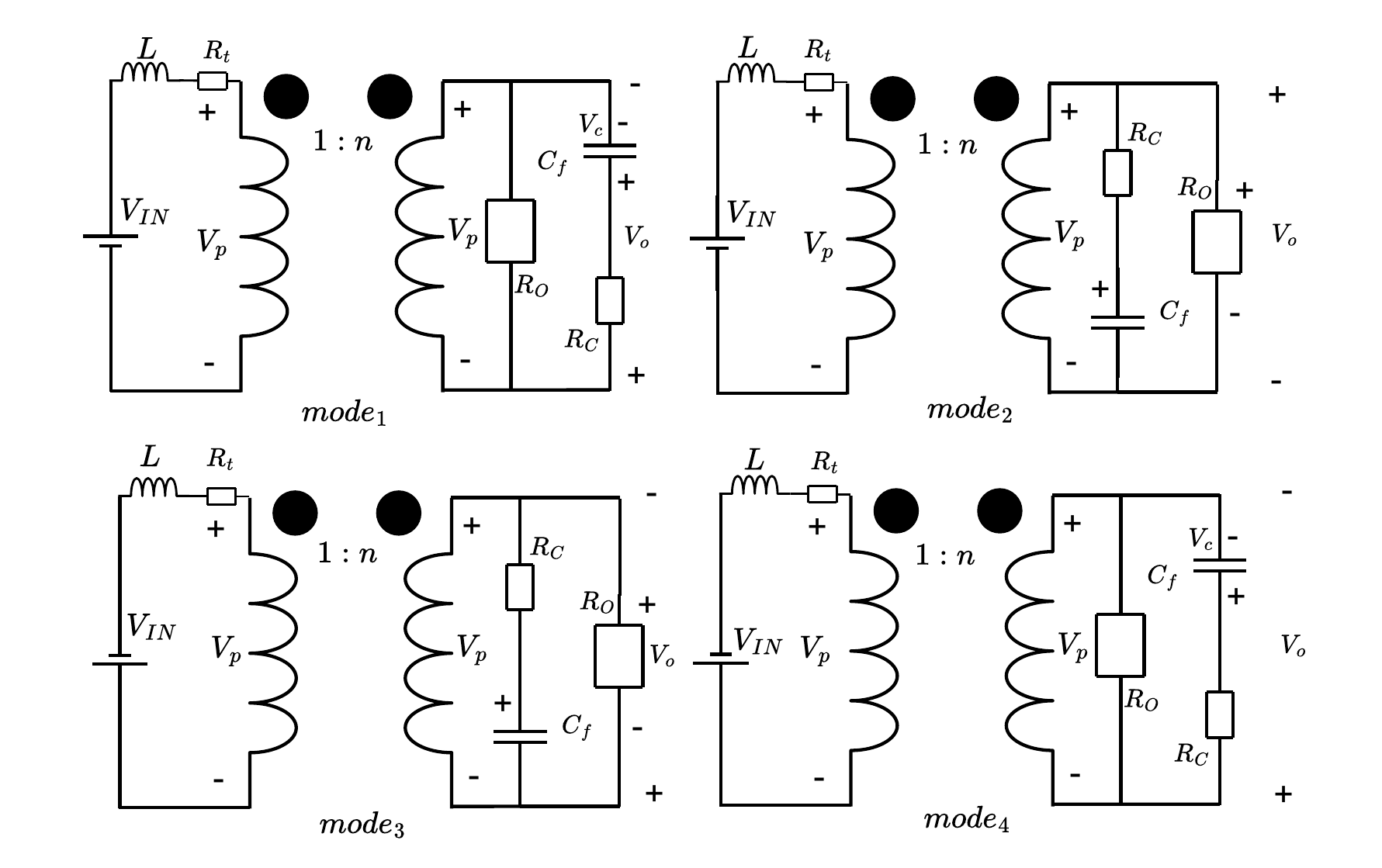}
  \caption{Schematic of the piecewise intervals.}
\end{figure}
\begin{figure*}[!t]
  \centering
  \captionsetup[subfigure]{justification=centering}
  \begin{tabular}{@{}c@{\;\;}c@{}}
    \begin{subfigure}[t]{0.49\textwidth}
      \includegraphics[width=\linewidth]{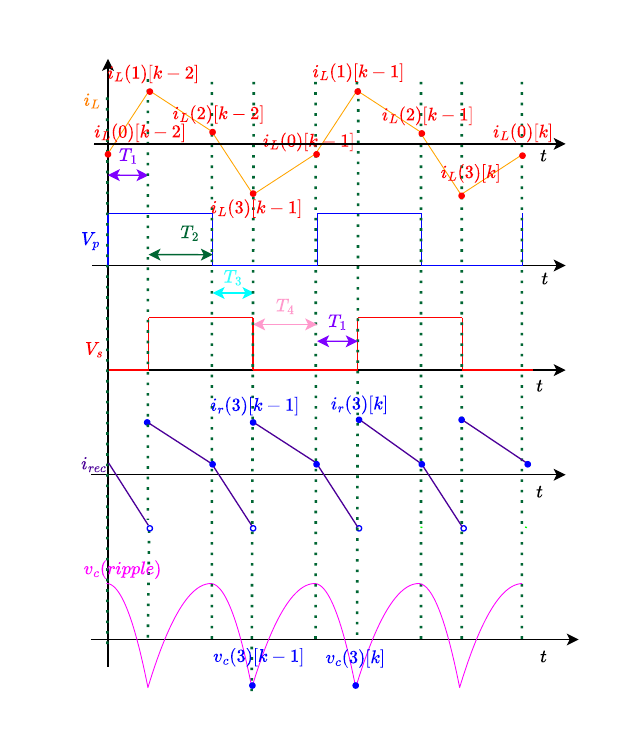}
      \caption{Primary leading-edge}\label{fig:1a}
    \end{subfigure}
    &
    \begin{subfigure}[t]{0.49\textwidth}
      \includegraphics[width=\linewidth]{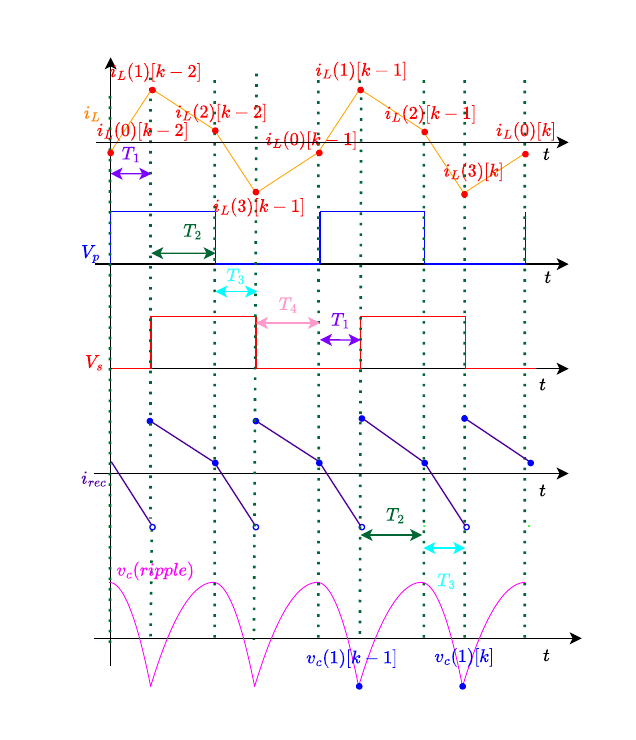}
      \caption{Primary trailing-edge}\label{fig:1b}
    \end{subfigure}
    \\[6pt]
    \begin{subfigure}[t]{0.49\textwidth}
      \includegraphics[width=\linewidth]{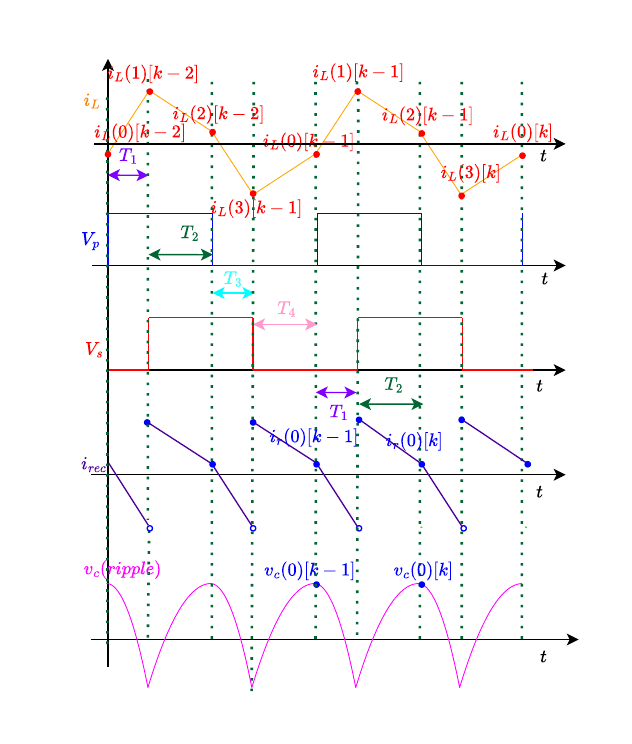}
      \caption{Secondary leading-edge}\label{fig:1c}
    \end{subfigure}
    &
    \begin{subfigure}[t]{0.49\textwidth}
      \includegraphics[width=\linewidth]{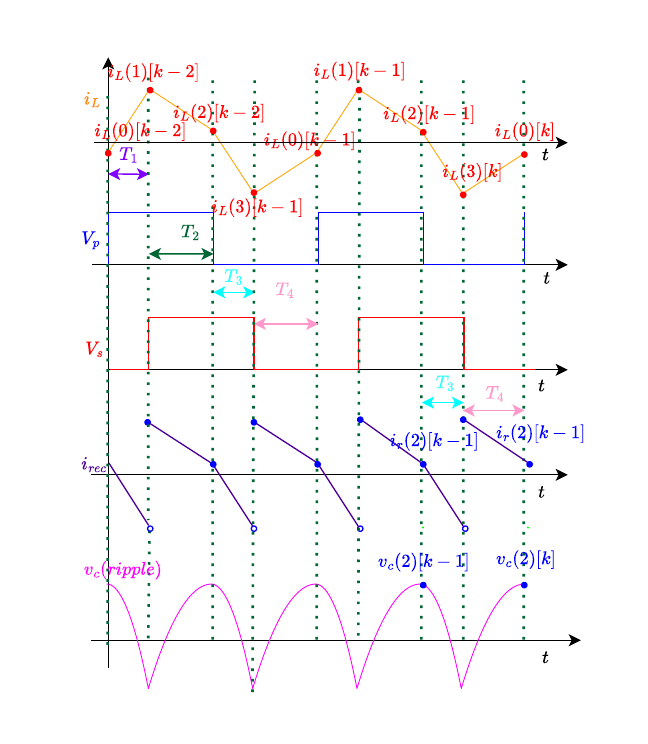}
      \caption{Secondary trailing-edge}\label{fig:1d}
    \end{subfigure}
  \end{tabular}
  \caption{Four modulation edges in DAB control.}
  \label{fig:four_panel}
\end{figure*}

\begin{enumerate}
    \item Subinterval $T_{1}$: $S_1$, $S_4$, $S_6$, $S_7$ are ON; $S_2$, $S_3$, $S_5$, $S_8$ are OFF($A_1$,$B_1$,$C_1$).
    \item Subinterval $T_{2}$: $S_1$, $S_4$, $S_5$, $S_8$ are ON; $S_2$, $S_3$, $S_6$, $S_7$ are OFF($A_2$,$B_2$,$C_2$).
    \item Subinterval $T_{3}$: $S_2$, $S_3$, $S_5$, $S_8$ are ON; $S_1$, $S_4$, $S_6$, $S_7$ are OFF.
    \item Subinterval $T_{4}$: $S_2$, $S_3$, $S_6$, $S_7$ are ON; $S_1$, $S_4$, $S_5$, $S_8$ are OFF.
\end{enumerate}

\begin{equation}
    X =
    \begin{bmatrix}
        i_L\\[4pt]v_{C_O}
    \end{bmatrix}
    \qquad
    Y =
    \begin{bmatrix}
        v_O
    \end{bmatrix}
    \qquad
    U =
    \begin{bmatrix}
        \overline{V_{IN}}
    \end{bmatrix}
\end{equation}
The output variable is : 
\[\mathbf{y}=
\begin{bmatrix}
I_{rec} \\
V_{out}
\end{bmatrix}\]

All matrices refer to the four‐interval model (\(n=4\)):

\fbox{\(A_1 = A_4\)}
\[
A_1 = A_4 =
\begin{bmatrix}
-\dfrac{n^2R_t + \dfrac{R_oR_C}{R_o+R_C}}{n^2L}
& \dfrac{R_o}{nL(R_o+R_C)} \\[8pt]
-\dfrac{R_o}{nC_o(R_o+R_C)}
& -\dfrac{1}{C_o(R_o+R_C)}
\end{bmatrix}
\]

\vspace{1em}

\fbox{\(A_2 = A_3\)}
\[
A_2 = A_3 =
\begin{bmatrix}
-\dfrac{n^2R_t + \dfrac{R_oR_C}{R_o+R_C}}{n^2L}
& -\dfrac{R_o}{nL(R_o+R_C)} \\[8pt]
\;\dfrac{R_o}{nC_o(R_o+R_C)}
& -\dfrac{1}{C_o(R_o+R_C)}
\end{bmatrix}
\]

\vspace{1em}

\fbox{\(B_1 = B_2\)}, \quad \fbox{\(B_3 = B_4\)}
\[
B_1 = B_2 = \begin{bmatrix} \dfrac{1}{L}\\[4pt]0 \end{bmatrix},
\qquad
B_3 = B_4 = \begin{bmatrix} -\dfrac{1}{L}\\[4pt]0 \end{bmatrix}.
\]

\[
\begin{aligned}
C_1 &= \begin{bmatrix}
-\,\dfrac{1}{n}\,
& 0\\
-\,\dfrac{1}{n}\,\bigl(R_C \parallel R_o\bigr)
& \dfrac{R_o}{R_C + R_o}
\end{bmatrix},\\[6pt]
C_2 &= \begin{bmatrix}
\,\dfrac{1}{n}\,
& 0\\\; \dfrac{1}{n}\,\bigl(R_C \parallel R_o\bigr)
& \dfrac{R_o}{R_C + R_o}
\end{bmatrix},\\[4pt]
C_3 &= C_2,\\
C_4 &= C_1,
\end{aligned}
\]
\section{Notation}
\[
S=\mathrm{diag}(1,-1),\qquad 
D'=\mathrm{diag}(-1,1),\qquad
SD'=-I .
\]
For each subinterval $i=1,\dots,4$ let
\[
\Phi_i=e^{A_iT_i},\qquad 
\Gamma_i=A_i^{-1}\bigl(\Phi_i-I\bigr)B_iU .
\]

\vspace{-0.4em}
\section{Similarity / Sign Identities}

\[
\begin{aligned}
A_2&=A_3=SA_1S,\quad &&A_4=A_1,\\
B_1&=B_2,\quad &&B_3=B_4=-B_1,\\
T_1&=T_3,\quad &&T_2=T_4 .
\end{aligned}
\]
\begin{lemma}[Similarity and Sign Relations]
Let
\[
\Phi_i = e^{A_i T_i},\quad
\Gamma_i = A_i^{-1}(\Phi_i - I)\,B_i\,U,
\quad
S = \diag(1,-1),
\]
with the structural and timing assumptions
\[
\begin{aligned}
A_2=A_3 &= S\,A_1\,S,\quad A_4=A_1,\\
B_1=B_2,\;B_3=B_4 &=-B_1,\quad
T_1=T_3,\;T_2=T_4.
\end{aligned}
\]
Then the following hold:
\[
\Phi_3 = S\,\Phi_1\,S,\quad
\Phi_4 = S\,\Phi_2\,S,
\quad
\Gamma_3 = -\,S\,\Gamma_1,\quad
\Gamma_4 = -\,S\,\Gamma_2.
\]
\end{lemma}

\begin{proof}
We verify each identity in turn, keeping every factor of $S$ explicit.

\medskip
\noindent\textbf{1. $\Phi_3 = S\,\Phi_1\,S$.}
\begin{equation}
\begin{aligned}
\Phi_3
&= e^{A_3 T_3}
= e^{(S A_1 S)\,T_1}
= S\,e^{A_1 T_1}\,S\\
&= S\,\Phi_1\,S.
\end{aligned}
\end{equation}

\medskip
\noindent\textbf{2. $\Phi_4 = S\,\Phi_2\,S$.}
\begin{equation}
\begin{aligned}
\Phi_4
&= e^{A_4 T_4}
= e^{A_1 T_2}
= S\,e^{A_1 T_2}\,S
= S\,\Phi_2\,S.
\end{aligned}
\end{equation}

\medskip
\noindent\textbf{3. $\Gamma_3 = -\,S\,\Gamma_1$.}
\begin{equation}
\begin{aligned}
\Gamma_3
&= A_3^{-1}(\Phi_3 - I)\,B_3\,U\\
&= (S A_1 S)^{-1}\,(S\Phi_1S - I)\,(-B_1)\,U\\
&= S A_1^{-1} S \;\bigl[S(\Phi_1 - I)S\bigr]\;(-B_1)\,U\\
&= -\,S\bigl[A_1^{-1}(\Phi_1 - I)\,B_1\,U\bigr]
= -\,S\,\Gamma_1.
\end{aligned}
\end{equation}

\medskip
\noindent\textbf{4. $\Gamma_4 = -\,S\,\Gamma_2$.}
\begin{equation}
\begin{aligned}
\Gamma_4
&= A_4^{-1}(\Phi_4 - I)\,B_4\,U\\
&= A_1^{-1}(\Phi_4 - I)(-B_1)\,U\\
&= -\,A_1^{-1}\bigl(S\Phi_2S - I\bigr)\,B_1\,U\\
&= -\,A_1^{-1}S\,(\Phi_2 - I)\,S\,B_1\,U\\
&\quad\text{(but }S B_1 = B_1\text{)}\\
&= -\,S\bigl[A_1^{-1}(\Phi_2 - I)\,B_1\,U\bigr]
= -\,S\,\Gamma_2.
\end{aligned}
\end{equation}

This completes the proof of all four identities.
\end{proof}

\section{$H$–$G$ Operator Proof for the Four-Step Fixed Point}
\begin{theorem}[Half‐Cycle Fixed‐Point Characterization]
Under the symmetry assumptions
\[
\begin{gathered}
S=\diag(1,-1),\quad D'=\diag(-1,1),\\ S\,D'=-I,
T_1=T_3,\;T_2=T_4,
\\A_2=A_3=S\,A_1\,S,\;\\A_4=A_1,\quad
B_3=B_4=-B_1,
\end{gathered}
\]
let
\[
\Phi_i=e^{A_iT_i},\quad
\Gamma_i=\int_0^{T_i}e^{A_i(T_i-\tau)}B_iU\,d\tau,
\]
and define the half‐cycle map
\[
H(X)=\Phi_2\Phi_1\,X+(\Phi_2\Gamma_1+\Gamma_2).
\]
Then the full‐period fixed‐point condition
\begin{equation}
    \begin{aligned}
    &X_4=\Phi_4\Phi_3\Phi_2\Phi_1\,X_0
           +\Phi_4\Phi_3\Phi_2\Gamma_1\\
          & +\Phi_4\Phi_3\Gamma_2
           +\Phi_4\Gamma_3
           +\Gamma_4
           =X_0
    \end{aligned}
\end{equation}
is equivalent to the much shorter “half‐cycle” condition
\[
H(X_0)=D'\,X_0.
\]
In other words, verifying
\[
\Phi_2\Phi_1\,X_0+(\Phi_2\Gamma_1+\Gamma_2)=D'\,X_0
\]
alone guarantees \(X_4=X_0\), thereby reducing the four‐step iteration to a two‐step check.
\end{theorem}
\begin{proof}
Step1: Define Two Half-Cycle Maps
\[
\begin{aligned}
H(X)&=\Phi_2\Phi_1X+\Phi_2\Gamma_1+\Gamma_2,\\
G(Z)&=\Phi_4\Phi_3Z+\Phi_4\Gamma_3+\Gamma_4 .
\end{aligned}
\]

Step2: Two-Step Flip Condition
\[
H(X_0)=D'X_0 .
\tag{H}
\]

Step3: Structure of $G$
\begin{equation}
\label{eq:Gcompact}
\begin{aligned}
G(Z)
&=S\Bigl[\,
      \Phi_2\Phi_1\bigl(SZ\bigr)
      -\Phi_2\Gamma_1-\Gamma_2
      \Bigr].
\end{aligned}
\end{equation}

Step4: Evaluate $G$ at $Z=H(X_0)$

Because $SD'=-I$,
\[
S\,H(X_0)=-X_0 .
\]

Insert into \eqref{eq:Gcompact}:
\begin{equation}
\label{eq:GofH}
\begin{aligned}
G\!\bigl(H(X_0)\bigr)
&=S\Bigl[
     \Phi_2\Phi_1(-X_0)
     -\Phi_2\Gamma_1-\Gamma_2
   \Bigr]\\[4pt]
&=-S\Bigl[
      \Phi_2\Phi_1X_0
      +\Phi_2\Gamma_1
      +\Gamma_2
   \Bigr].
\end{aligned}
\end{equation}

Step5: Apply Two-Step Condition
Using (H),
\[
\Phi_2\Phi_1X_0+\Phi_2\Gamma_1+\Gamma_2=D'X_0 .
\]

Substitute into \eqref{eq:GofH}:
\[
G\!\bigl(H(X_0)\bigr)=-S\,D'X_0=X_0 .
\]

Step6: Four-Step Closure
Since $X_2=H(X_0)$,
\[
X_4
  =G(X_2)
  =G\!\bigl(H(X_0)\bigr)
  =X_0 .
\]
\end{proof}

\section{Small-Signal Model of DAB Under Fixed-Frequency Phase Modulation}
\label{sec:dab_ff_small_signal}

\subsection{Notation, rectified half-cycle state, and SIMO physical output}
\label{subsec:dab_notation_simo}

We consider the four-subinterval DAB model over one switching period $T_s$,
and define the half-cycle length $T_h:=T_s/2$.
The state is
\begin{equation}
x(t)=
\begin{bmatrix}
i_L(t)\\[2pt]
v_C(t)
\end{bmatrix}.
\end{equation}
On subinterval $i\in\{1,2,3,4\}$, the dynamics are
\begin{equation}
\dot x(t)=A_i\,x(t)+B_i\,U,
\end{equation}
where $U$ is constant over each subinterval.

\subsubsection{A. Rectified half-cycle coordinate}
To make the sampled inductor-current sign-consistent across half-cycles,
we use the fixed involution
\begin{equation}
D_r:=\mathrm{diag}(-1,\,1),
\qquad
D_r^{-1}=D_r.
\label{eq:Dr_def_sec}
\end{equation}
All half-cycle maps below are written in this rectified coordinate.

\subsubsection{B. Exact segment maps}
For subinterval $i$ with duration $T_i$, define the exact segment map
\begin{equation}
x^{+}=\Phi_i x^{-}+\Gamma_i,
\Phi_i:=e^{A_iT_i},
\Gamma_i:=\int_0^{T_i}e^{A_i(T_i-\tau)}B_iU\,d\tau .
\label{eq:PhiGamma_def_sec}
\end{equation}

\subsubsection{C. SIMO physical output matrix}
We adopt the SIMO physical output
\begin{equation}
y_k=
\begin{bmatrix}
I_{\mathrm{rec},k}\\[2pt]
V_{\mathrm{out},k}
\end{bmatrix}
=C_{\mathrm{phys}}\,x_k,
\label{eq:y_Cphys_def_sec}
\end{equation}
where $C_{\mathrm{phys}}$ absorbs the turns ratio and the output definition,
so no selection matrix is needed:
\begin{equation}
C_{\mathrm{phys}}:=
\begin{bmatrix}
\frac{1}{n} & 0\\[6pt]
\frac{R_c\parallel R_o}{n} & \frac{R_o}{R_c+R_o}
\end{bmatrix},
\qquad
R_c\parallel R_o:=\frac{R_cR_o}{R_c+R_o}.
\label{eq:Cphys_def_sec}
\end{equation}
Importantly, $C_{\mathrm{phys}}$ is \emph{the same physical mapping}
from $(i_L,v_C)$ to $(I_{\mathrm{rec}},V_{\mathrm{out}})$ on \emph{every}
sampling surface.

\subsection{DAB symmetry identities}
\label{subsec:dab_symmetry}

Let $S:=\mathrm{diag}(1,-1)$ so that
\begin{equation}
D_r=-S.
\label{eq:Dr_minus_S}
\end{equation}
Under the four-interval DAB structure (matched half-cycles), we use
\begin{equation}
\begin{aligned}
&A_4=A_1,\qquad A_2=A_3=S A_1 S,\\
&B_1=B_2,\qquad B_3=B_4=-B_1,\\
&T_1=T_3,\qquad T_2=T_4 .
\end{aligned}
\label{eq:DAB_structural_identities_sec}
\end{equation}
Consequently, the state-transition matrices satisfy
\begin{equation}
\Phi_3=D_r\Phi_1D_r,\qquad \Phi_4=D_r\Phi_2D_r.
\label{eq:Phi_conjugacy_sec}
\end{equation}

\subsection{Two-interval rectified half-cycle map on an arbitrary surface}
\label{subsec:two_interval_template}

Pick any half-cycle sampling surface such that the half-cycle evolution
consists of two consecutive subintervals $a\rightarrow b$ with durations
$T_a$ and $T_b$ satisfying
\begin{equation}
T_a+T_b=T_h.
\label{eq:halfcycle_sum}
\end{equation}

\subsubsection{A. Large-signal half-cycle map}
The rectified half-cycle map is
\begin{equation}
\begin{aligned}
x_{k+1}
&=D_r\Big(
\Phi_b\big(\Phi_a x_k+\Gamma_a\big)+\Gamma_b
\Big)\\
&=: \Phi^{(ab)}x_k+g^{(ab)},
\end{aligned}
\label{eq:halfcycle_map_sec}
\end{equation}
where
\begin{equation}
\Phi^{(ab)}:=D_r\Phi_b\Phi_a,
\qquad
g^{(ab)}:=D_r(\Phi_b\Gamma_a+\Gamma_b).
\label{eq:Phiab_gab_def_sec}
\end{equation}

\subsubsection{B. Endpoint timing sensitivities}
Let $x^\star$ be the fixed point of \eqref{eq:halfcycle_map_sec}, and define
the intermediate steady states
\begin{equation}
\begin{aligned}
x_{a,\mathrm{end}}^\star&=\Phi_a x^\star+\Gamma_a,\\
x_{b,\mathrm{end}}^\star&=\Phi_b x_{a,\mathrm{end}}^\star+\Gamma_b.
\end{aligned}
\label{eq:intermediate_states_sec}
\end{equation}
Using the endpoint vector-field identity, the duration sensitivities are
\begin{equation}
\begin{aligned}
\eta_a^{(ab)}
&:=\left.\frac{\partial x_{k+1}}{\partial T_a}\right|_{\star}
= D_r\,\Phi_b\Big(A_a x_{a,\mathrm{end}}^\star+B_aU\Big),\\[6pt]
\eta_b^{(ab)}
&:=\left.\frac{\partial x_{k+1}}{\partial T_b}\right|_{\star}
= D_r\Big(A_b x_{b,\mathrm{end}}^\star+B_bU\Big).
\end{aligned}
\label{eq:eta_ab_def_sec}
\end{equation}

\subsection{Fixed-frequency modulation: correct two-cycle timing split}
\label{subsec:ff_correct_twocycle}

A fixed-frequency phase modulator anchors one endpoint to the clock and
generates the other endpoint by a comparator against a ramp of amplitude $V_r$.
Define the ramp slope magnitude and its reciprocal as
\begin{equation}
|S_e|=\frac{V_r}{T_h},
\qquad
\kappa:=\frac{1}{|S_e|}=\frac{T_h}{V_r}.
\label{eq:kappa_def_sec}
\end{equation}

\subsubsection{A. Timing law and polarity (primary vs.\ secondary)}
Because the half-cycle partition satisfies $T_a+T_b=T_h$,
the two adjacent timing perturbations are always opposite in sign.
We write the fixed-frequency two-cycle timing law as
\begin{equation}
\widehat T_{a,k}= \rho\,\kappa\,\hat v_{c,k},
\qquad
\widehat T_{b,k}= -\rho\,\kappa\,\hat v_{c,k+1},
\label{eq:timing_law_rho_sec}
\end{equation}
where $\rho\in\{+1,-1\}$ captures the modulator logic:
\begin{equation}
\rho=
\begin{cases}
+1, & \text{secondary-side phase modulation (pos logic)},\\
-1, & \text{primary-side phase modulation (neg logic)}.
\end{cases}
\label{eq:rho_table_sec}
\end{equation}

\subsubsection{B. Linearized half-cycle map}
Linearizing \eqref{eq:halfcycle_map_sec} and substituting
\eqref{eq:timing_law_rho_sec} yield
\begin{equation}
\hat x_{k+1}
=\Phi^{(ab)}\hat x_k
+\beta_-^{(ab)}\hat v_{c,k}
+\beta_+^{(ab)}\hat v_{c,k+1},
\label{eq:lin_map_beta_sec}
\end{equation}
where we absorb $\kappa$ into the two vectors
\begin{equation}
\beta_-^{(ab)}:=\rho\,\kappa\,\eta_a^{(ab)},
\qquad
\beta_+^{(ab)}:=-\rho\,\kappa\,\eta_b^{(ab)}.
\label{eq:beta_def_sec}
\end{equation}

\subsubsection{C. SIMO $z$-domain transfer function}
Taking the $z$ transform of \eqref{eq:lin_map_beta_sec} gives
\begin{equation}
\hat X(z)
=(zI-\Phi^{(ab)})^{-1}\big(\beta_-^{(ab)}+z\beta_+^{(ab)}\big)\hat V_c(z).
\end{equation}
Therefore, the SIMO transfer from $\hat v_c$ to
$y=[I_{\mathrm{rec}},V_{\mathrm{out}}]^T$ is
\begin{equation}
\boxed{
H^{(ab)}_{\mathrm{fix}}(z)
=
C_{\mathrm{phys}}\,(zI-\Phi^{(ab)})^{-1}
\big(\beta_-^{(ab)}+z\beta_+^{(ab)}\big).
}
\label{eq:H_fix_general_sec}
\end{equation}

\subsection{Four sampling surfaces and their transfer functions}
\label{subsec:four_surfaces}

The four-interval DAB admits four natural half-cycle sampling surfaces:
\begin{equation}
\begin{aligned}
&\mathrm{P}^+:(a,b)=(1,2),\quad
\mathrm{S}^+:(a,b)=(2,3),\quad
\\&\mathrm{P}^-:(a,b)=(3,4),\quad
\mathrm{S}^-:(a,b)=(4,1).
\label{eq:four_surfaces_def_sec}
\end{aligned}
\end{equation}
Each surface yields \eqref{eq:H_fix_general_sec}, with the corresponding
$\Phi^{(ab)}$, $\eta_a^{(ab)}$, $\eta_b^{(ab)}$, and polarity $\rho$
given by \eqref{eq:rho_table_sec}.

\subsection{Primary-side and secondary-side equivalence via similarity}
\label{subsec:pri_sec_equivalence}

This subsection formalizes the key fact you validated: even with a single
physical $C_{\mathrm{phys}}$, the primary-side and secondary-side realizations
are related by a similarity map between \emph{surface states}. Hence they
represent the same physical input--output transfer.

\subsubsection{A. A similarity identity for resolvents}
\begin{lemma}
\label{lem:resolvent_similarity}
Let $T$ be nonsingular. For any square matrix $A$,
\begin{equation}
(zI-T^{-1}AT)^{-1}=T^{-1}(zI-A)^{-1}T.
\label{eq:resolvent_similarity_id}
\end{equation}
\end{lemma}

\subsubsection{B. Key similarity: $\Phi^{(12)}$ and $\Phi^{(23)}$}
Using \eqref{eq:Phi_conjugacy_sec},
\begin{equation}
\Phi^{(23)}
=D_r\Phi_3\Phi_2
=D_r(D_r\Phi_1D_r)\Phi_2
=\Phi_1D_r\Phi_2,
\label{eq:Phi23_factor_sec}
\end{equation}
and
\begin{equation}
\Phi^{(12)}=D_r\Phi_2\Phi_1.
\label{eq:Phi12_factor_sec}
\end{equation}
Choose the similarity transform
\begin{equation}
T:=\Phi_1.
\label{eq:T_choice_sec}
\end{equation}
Then
\begin{equation}
\begin{aligned}
T^{-1}\Phi^{(23)}T
&=\Phi_1^{-1}\big(\Phi_1D_r\Phi_2\big)\Phi_1\\
&=(\Phi_1^{-1}\Phi_1)D_r\Phi_2\Phi_1\\
&=D_r\Phi_2\Phi_1=\Phi^{(12)}.
\end{aligned}
\label{eq:Phi_similarity_sec}
\end{equation}

\subsubsection{C. Matching the $z$-weighted input vector with correct polarity}
Define the $z$-weighted input vector
\begin{equation}
b^{(ab)}(z):=\beta_-^{(ab)}+z\beta_+^{(ab)}.
\label{eq:bz_def_sec}
\end{equation}
For $\mathrm{P}^+$ we use $\rho=-1$; for $\mathrm{S}^+$ we use $\rho=+1$
as in \eqref{eq:rho_table_sec}.
With this primary/secondary polarity distinction, the corresponding
$z$-weighted vectors satisfy the surface-coordinate relation
\begin{equation}
b^{(12)}(z)=T^{-1}b^{(23)}(z),
\qquad T=\Phi_1.
\label{eq:bz_similarity_sec}
\end{equation}
(If one incorrectly uses the same polarity for both sides,
a global sign mismatch appears.)

\subsubsection{D. Transfer equivalence in a common surface coordinate}
Let $x^{(12)}_k$ denote the surface state on $\mathrm{P}^+$ and
$x^{(23)}_k$ the surface state on $\mathrm{S}^+$.
From the surface relation induced by $T=\Phi_1$, we use
\begin{equation}
x^{(12)}_k=T^{-1}x^{(23)}_k.
\label{eq:surface_state_relation}
\end{equation}
Therefore, the output mapping in the $\mathrm{S}^+$ coordinate becomes
\begin{equation}
y_k=C_{\mathrm{phys}}x^{(12)}_k
=C_{\mathrm{phys}}T^{-1}x^{(23)}_k.
\label{eq:C_transform_surface}
\end{equation}
This is the precise sense in which ``a single $C_{\mathrm{phys}}$'' is
consistent with similarity between surface states: the physical mapping
is the same, but expressing it in a different surface coordinate inserts
the $T^{-1}$ factor.

Now start from the $\mathrm{P}^+$ transfer:
\begin{equation}
H_{\mathrm{P}^+}(z)
=C_{\mathrm{phys}}(zI-\Phi^{(12)})^{-1}b^{(12)}(z).
\label{eq:H_Pplus_def}
\end{equation}
Substitute \eqref{eq:Phi_similarity_sec} and \eqref{eq:bz_similarity_sec}:
\begin{equation}
\begin{aligned}
H_{\mathrm{P}^+}(z)
&=C_{\mathrm{phys}}
\Big(zI-T^{-1}\Phi^{(23)}T\Big)^{-1}T^{-1}b^{(23)}(z)\\
&=C_{\mathrm{phys}}
\Big(T^{-1}(zI-\Phi^{(23)})^{-1}T\Big)T^{-1}b^{(23)}(z)\\
&=\big(C_{\mathrm{phys}}T^{-1}\big)(zI-\Phi^{(23)})^{-1}b^{(23)}(z).
\end{aligned}
\label{eq:H_Pplus_to_Splus_coordinate}
\end{equation}
By \eqref{eq:C_transform_surface}, the right-hand side is exactly the
$\mathrm{S}^+$ transfer expressed in the $\mathrm{S}^+$ surface coordinate.
Hence the two surfaces are physically input--output equivalent.

Applying the same argument to the other pair $(34)\leftrightarrow(41)$
yields the equivalence between $\mathrm{P}^-$ and $\mathrm{S}^-$.

\subsection{Strong same-cycle constraint and why it is accurate at most frequency}
\label{subsec:strong_constraint_lowfreq}

A common simplification replaces $\hat v_{c,k+1}$ by $\hat v_{c,k}$ in
\eqref{eq:lin_map_beta_sec}, yielding a single-cycle input model.

\subsubsection{A. Strong same-cycle model}
Approximating $\hat v_{c,k+1}\approx \hat v_{c,k}$ gives
\begin{equation}
\hat x_{k+1}
=\Phi^{(ab)}\hat x_k+\beta_{\mathrm{sc}}^{(ab)}\hat v_{c,k},
\qquad
\beta_{\mathrm{sc}}^{(ab)}:=\beta_-^{(ab)}+\beta_+^{(ab)}.
\label{eq:strong_constraint_model_sec}
\end{equation}
Thus
\begin{equation}
H_{\mathrm{sc}}^{(ab)}(z)
=
C_{\mathrm{phys}}(zI-\Phi^{(ab)})^{-1}\beta_{\mathrm{sc}}^{(ab)}.
\label{eq:H_sc_def_sec}
\end{equation}

\subsubsection{B. Exact--approximate difference and low-frequency bound}
Subtracting \eqref{eq:H_sc_def_sec} from \eqref{eq:H_fix_general_sec} yields
\begin{equation}
\begin{aligned}
\Delta H^{(ab)}(z)
&:=H_{\mathrm{fix}}^{(ab)}(z)-H_{\mathrm{sc}}^{(ab)}(z)\\
&=C_{\mathrm{phys}}(zI-\Phi^{(ab)})^{-1}(z-1)\beta_+^{(ab)}.
\end{aligned}
\label{eq:DeltaH_def_sec}
\end{equation}
For $z=e^{j\omega T_h}$,
\begin{equation}
|z-1|=2\left|\sin\!\left(\frac{\omega T_h}{2}\right)\right|
\le \omega T_h,
\qquad \omega T_h\ll 1.
\label{eq:zminus1_bound_sec}
\end{equation}
If $\Phi^{(ab)}$ is stable and has no eigenvalue close to $1$,
then $(zI-\Phi^{(ab)})^{-1}$ remains bounded near $z\approx 1$ and
\eqref{eq:DeltaH_def_sec} implies
\begin{equation}
\big\|\Delta H^{(ab)}(e^{j\omega T_h})\big\|
=O(\omega T_h),
\qquad \omega\to 0.
\label{eq:lowfreq_error_order_sec}
\end{equation}

\begin{remark}
Equation \eqref{eq:DeltaH_def_sec} shows the discrepancy is shaped by a
discrete-difference factor $(z-1)$, which explains why the strong same-cycle
model can be nearly indistinguishable from the correct fixed-frequency model
over a wide band, especially when the rectified half-cycle plant has no
marginal pole at $z=1$.
\end{remark}

\bibliographystyle{IEEEtran}
\bibliography{IEEEabrv,biblio}

\end{document}